\documentclass[10ps]{article}
\usepackage{amsmath}
\usepackage{amsfonts}
\usepackage{amssymb}
\usepackage{amscd}
\usepackage{amsthm}
\usepackage{latexsym}
\usepackage{mathrsfs}
\usepackage[all]{xy}
\usepackage{tikz}
\usepackage{tikz-cd}
\usepackage{color}
\usepackage{stmaryrd}
\usepackage{picture}
\usepackage{fullpage}
\usepackage{authblk}
\usepackage[colorinlistoftodos]{todonotes}

\newcommand{\id}{\mathit{id}}
\newcommand{\dom}{\mathit{dom}}
\newcommand{\cod}{\mathit{cod}}
\newcommand{\so}{\ell}
\newcommand{\ta}{r}
\newcommand{\The}{\rotatebox[origin=c]{180}{$\iota$}}
\newcommand{\Pow}{\mathcal{P}}

\newtheorem{theorem}{Theorem}[section]
\newtheorem{proposition}[theorem]{Proposition}
\newtheorem{lemma}[theorem]{Lemma}

\theoremstyle{definition}
\newtheorem{definition}[theorem]{Definition}

\definecolor{jccolor}{rgb}{1,0.5,0.5}
\definecolor{sdcolor}{cmyk}{1,0,1,0}
\definecolor{gscolor}{cmyk}{1,0,0,0}

\begin{document}

\title{Relational Semigroups and Object-Free Categories}

\author{James Cranch}
\author{Simon Doherty}
\author{Georg Struth}
\affil{University of Sheffield, UK}

\date{}

\maketitle

% \newtheorem{thm}{Theorem}
% \newtheorem{prop}[thm]{Proposition}
% \newtheorem{lem}[thm]{Lemma}
% \newtheorem{ex}[thm]{Example}
% \newtheorem{cor}[thm]{Corollary}

% \newdefinition{definition}{Definition}

% \newproof{pf}{Proof}

% \journal{Information Processing Letters}

% \begin{document}

% \begin{frontmatter}

% %\title{On Partial Monoids and  Object-Free Categories}
% \title{Relational Semigroups and Categories}

% \author[add1]{James Cranch}
% \ead{j.d.cranch@shef.ac.uk}
% \author[add2]{Simon Doherty}
% \ead{s.doherty@sheffield.ac.uk}
% \author[add2]{Georg Struth}
% \ead{g.struth@sheffield.ac.uk}
% \address[add1]{School of Mathematics and Statistics, University of Sheffield, UK}
% \address[add2]{Department of Computer Science, University of Sheffield, UK}
% %\email{g.struth@sheffield.ac.uk}
% \date{}

\begin{abstract}
  This note relates axioms for partial semigroups and monoids with
  those for small object-free categories, either with multiple
  monoidal units or with source and target maps. We discuss the
  adjunction of a zero element to both kinds of category and
  provide examples that separate the algebras considered.
\end{abstract}

% \begin{keyword}
%   Partial Monoids \sep Object-Free Categories \sep
%   Algebras of Functions
% \end{keyword}

%\end{frontmatter}

\pagestyle{plain}

%%%%%%%%%%%%%%%%%%%%%%%%%%%%%%%%%%%%%%%%%%%%%%%%%%%%%%%%%%%%%%%%%%%%%%%%%%%%%%

\section{Introduction}\label{S:introduction}

Partial semigroups and monoids have been studied widely in mathematics
and applied in computer science and quantum physics. Among the
examples are object-free or arrows-only categories. These can be
axiomatised in at least two different ways~\cite{MacLane98}: as
certain partial monoids with many units~\cite[p.9]{MacLane98} or as
partial semigroups with source and target
operations~\cite[p.279]{MacLane98}. An early example of the first type
of axiomatisation is due to Kurosh, Livshits and
Shul'geifer~\cite{KuroshLS60}. A precursor are Brandt
groupoids~\cite{Brandt27}, which generalise groups to a partial
setting. The second type of axiomatisation features also in Freyd and
Scedrov's book~\cite{FreydS90}. It is foreshadowed by Schweizer and
Sklar's work on function systems~\cite{SchweizerS67}.

We have previously used partial semigroups and monoids in various
lifting constructions~\cite{DongolHS16,DongolHS17}. This note intends
to clarify their relationship with categories.

We model partial operations of type $X\times X\to X$ as ternary
relations on $X$ that satisfy a functionality condition.  Type-one
categories then turn out to be partial monoids with multiple units
that satisfy a certain coherence condition, which corresponds to the
matching condition between source and target objects in the
composition of arrows in a category.  We outline how this relational
approach gives rise to the more conventional model of partial
operations as functions of type $D\to X$ (or pullbacks), where
$D\subseteq X\times X$ is the domain of definition of the partial
operation. Next we relate the approach to type-two object-free
categories. Alternatively, partial operations are often modelled in a
total setting obtained by adjoining a zero or bottom element. We
describe the effect of this adjunction on both types of partial
semigroups and object-free categories. Finally, we present examples
that separate the various partial algebras considered.

%%%%%%%%%%%%%%%%%%%%%%%%%%%%%%%%%%%%%%%%%%%%%%%%%%%%%%%%%%%%%%%%%%%%%%%%

\section{Properties of Ternary Relations}\label{S:preliminaries}

We start with a list of properties of a ternary relation $R\subseteq X\times Y\times Z$.
\begin{itemize}
\item $R$ is \emph{weakly functional} (\emph{functional}) if for all
  $y\in Y$ and $z\in Z$ there is at most (precisely) one $x\in X$ such
  that $R^x_{yz}$. A (weakly) functional ternary relation on $X$ is a
  (\emph{partial}) \emph{operation}.
\item $R$ is \emph{relationally associative} if
$\exists v.\ R^u_{xv} \land R^v_{yz} \Leftrightarrow \exists v.\
R^u_{vz}\land R^v_{xy}$.
\item $R$ is \emph{coherent} if
  $R^v_{xy} \land D^y_z\Rightarrow D^v_z$, where
  $D^x_y\Leftrightarrow \exists z.\ R^z_{xy}$. 
\item An element $e\in Y$ is
  a \emph{relational left unit} of $R$ if $\exists x\in X.\ R^x_{ex}$
  and $\forall x,y\in X.\ R^y_{ex} \Rightarrow y =x$. It is a
  \emph{relational right unit} of $R$ if
  $\exists x\in X.\ R^x_{xe}$ and
  $\forall x,y\in X.\ R^y_{xe} \Rightarrow y =x$.  It is a \emph{unit}
  of $R$ if it is either a left or a right unit of $R$. We write $E$
  for the set of units of $R$.
\end{itemize}

Partial operations as ternary relations date back to Brandt's work on
generalised groups~\cite{Brandt27}. Intuitively, relations
$R:X\to Y\to Z\to 2$ are isomorphic to functions
$f_R:Y\times Z \to \Pow\, X$. Then, for all $x,y\in X$, $|f_R\, x\, y|=1$ if
$R$ is an operation, and $|f_R\, x\, y|\le 1$ if $R$ is a partial
operation ($f_R\, x\, y=\emptyset\Leftrightarrow \neg D^x_y$).
\begin{itemize}
\item A \emph{morphism} between relations $R\subseteq X\times Y\times Z$ and
$R'\subseteq X'\times Y'\times Z'$ is a triple of functions
$f: X \to X'$, $g: Y \to Y'$ and $h: Z \to Z'$ that satisfies
\begin{equation*}
  R^x_{yz} \Rightarrow {R'}^{(f\, x)}_{(g\, y)(h\, z)}
\end{equation*}
\item It is \emph{bounded} if
${R'}^{f\, x}_{vw} \Rightarrow \exists y,z.\ (v=g\, y) \land (w=h\, z) \land R^x_{yz}$
\end{itemize}
We consider morphisms because they correspond to functors between
object-free categories. Bounded morphisms are interesting for
Stone-type dualities between ternary relations and quantales, see
Section~\ref{S:unit-collapse}. For relations $R \subseteq X \times X \times X$,
we consider morphisms of the form $(f, f, f)$ rather than the
more general $(f, g, h)$.

Partiality is often modelled in a total setting by adjoining a zero or
bottom element. 
\begin{itemize}
\item A \emph{relational zero} of $R$ is an element $0\in X$
that satisfies $R^0_{0x}$ and $R^0_{x0}$ for all $x\in X$.  

\end{itemize}
One can always adjoin a zero to a relational structure $(X,R)$ by
defining $X_0=X\cup\{0\}$ and extending $R$ to $R_0$ by adding the
triples $\left(R_0\right)^0_{xy}\Leftrightarrow \neg D^x_y$ as well as
$(R_0)^0_{0x}$, $(R_0)^0_{x0}$ and $(R_0)^0_{00}$.  This makes $R$
total: $\left(D_0\right)^x_y$ always holds.

%%%%%%%%%%%%%%%%%%%%%%%%%%%%%%%%%%%%%%%%%%%%%%%%%%%%%

\section{Relational Monoids and Type-One Categories}\label{S:relational-monoids}

\begin{definition}\label{D:rel-semigroup}
  A \emph{relational semigroup} is relational structure $(X,R)$ 
  with a relationally associative relation
  $R\subseteq X\times X\times X$.  It is a \emph{(partial) semigroup} if
  $R$ is a (partial) operation.
\end{definition}
\begin{definition}
  A \emph{relational monoid} is a relational semigroup $(X,R)$ in
  which for every $x\in X$ there are $e,e'\in E$ such that $D^e_x$ and
  $D^x_{e'}$. It is a \emph{(partial) monoid} if $R$ is a (partial)
  operation.
\end{definition}
\begin{definition}
  An \emph{object-free category} is a coherent partial monoid.
\end{definition}
Objects and identity arrows in a category are obviously in one-to-one
correspondence, so that the latter can always be used in place of the
former. This is the main idea behind object-free categories.

For partial semigroups and monoids, we can define a partial
composition by $x\cdot y = \The z.\ R^z_{xy}$ whenever $D^x_y$, where
$\The$ denotes the definite description operator. It follows that
$R^x_{yz}\Leftrightarrow D^y_z\land x= y\cdot z$.   The
standard axioms for small object-free categories can then be derived using
this algebraic notation.
\begin{lemma}\label{P:ahs-derivable}
In every object-free category, 
\begin{gather*}
  D^x_y \wedge D^{x\cdot y}_z \Leftrightarrow  D^x_{y\cdot z} \wedge
  D^y_z,\qquad
 D^x_y \land D^{x\cdot  y}_z \Rightarrow  (x \cdot y) \cdot z = x
 \cdot (y \cdot z),\\
\forall x\exists e\in E.\ D^e_x, \qquad \forall x\exists e\in E.\ D^x_e,\\
D^x_y \land D^y_z \Rightarrow D^{x\cdot y}_z.
\end{gather*}
\end{lemma}
The laws in the first line are standard axioms for partial
semigroups. Together with those in the second line they axiomatise
partial monoids with many units (Dongol, Hayes and
Struth~\cite{DongolHS16} use different, but equivalent unit
axioms). The law in the third line is the algebraic version of
coherence.

In relational monoids, it is easy to show that units can always be
composed with themselves, but never with other units. Moreover, every
element has precisely one left and one right unit.  This functional
dependency is captured by source and target functions $\ell,r:S\to S$
defined as
\begin{equation*}
  \so\, x = \The e\in E.\ R^x_{ex}\qquad \text{ and }\qquad \ta\, x =
  \The e\in E.\ R^x_{xe}.
\end{equation*}
\begin{lemma}\label{P:rel-props3}
  Let $X$ be a relational monoid and $e,e'\in E$. Then, for all $x,y\in X$,
  \begin{enumerate}
\item $\ta\, (\so\, x) = \so\, x$ and $\so\, (\ta\, x) = \ta\, x$,
\item $R^x_{(\so\, x)x}$ and $R^x_{x(\ta\, x)}$\quad ($\so\, x \cdot x = x$
  and $x\cdot \ta\, x = x$),
\item $R^v_{xy}\land R^w_{x(\so\, y)} \Rightarrow \so\, v = \so\, w$
  and $R^v_{xy} \land R^w_{(\ta\, x)y} \Rightarrow \ta\, v= \ta\, w$\\
  ($D^x_y\Rightarrow \so\, (x\cdot y) = \so\, (x\cdot \so\, y)$
  and
  $D^x_y\Rightarrow \ta\, (x\cdot y) = \ta\, (\ta\, x\cdot y)$),
\item $R^v_{xy}\Rightarrow \so\, v = \so\, x$ and
  $R^v_{xy} \Rightarrow \ta\, v= \ta\, y$\quad
  ($D^x_y\Rightarrow \so\, (x\cdot y) =\so\, x$ and
  $D^x_y\Rightarrow \ta\, (x\cdot y) =\ta\, y$),
\item $D^x_y \Rightarrow \left( R^v_{(\ta\,
    x) (\so\, y)} \Leftrightarrow  R^v_{(\so\, y) (\ta\, x)}\right)$\quad
($D^x_y\Rightarrow \so\, x \cdot \ta\, y=\ta\, y\cdot \so\, x$),
\item  $R^u_{x y} \Rightarrow \left( R^v_{(\ta\,  x) y}
  \Leftrightarrow 
  R^v_{y (\ta\,  u)}\right)$\quad
 ($D^x_y\Rightarrow \ta\, x \cdot y = y \cdot \ta\, (x\cdot
y)$).
\end{enumerate}
Formula in brackets indicate algebraic variants for partial monoids.
\end{lemma}
Up to definedness conditions, the statements in
Lemma~\ref{P:rel-props3}, except (4), are Schweizer and Sklar's axioms
for function systems~\cite{SchweizerS67}; (4) obviously holds in any
category.

Source and target functions bring coherence closer to categorical
intuitions.
\begin{lemma}\label{P:coherent-rel-monoid}
 In any relational monoid $(X,R)$, 
 \begin{equation*}
 R\text{ is coherent}
\Leftrightarrow \left(\forall x,y\in X. \forall e\in E.\ D^x_e\land
  D^e_y\Rightarrow D^x_y\right)\Leftrightarrow \left(\forall x,y\in X.\ \ta\, x =\so\, y\Rightarrow
  D^x_y\right).
 \end{equation*}
\end{lemma}
In any coherent relational monoid, in fact,
$D^x_y \Leftrightarrow \exists e\in E.\ D^x_e\land D^e_y
\Leftrightarrow \ta\, x = \so\, y$
because $D^x_y\Rightarrow \ta\, x =\so\, y$ holds already in
relational monoids. 

Partial semigroups, partial monoids and object-free categories are
usually defined slightly differently as algebraic structures
$(S,\cdot,D)$ in which composition is a partial operation of type
$D\to S$ for $D\subseteq S\times S$. The appropriate laws of
Lemma~\ref{P:ahs-derivable} are then imposed in this setting (plus the
obvious algebraic axioms defining the units in $E$ if needed).
Morphisms $f:S\to S'$ between such algebraic structures $S$ and $S'$
satisfy that $D^x_y$ implies ${D'}^{f\, x}_{f\, y}$ and
$f\, (x \cdot y) = f\, x \cdot' f\, y$. They are bounded if ${D'}^u_v$
and $f\, x = u\cdot' v$ imply that there exist $y$ and $z$ such that
$f\, y = u$, $f\, z= v$, $D^y_z$ and $x= y\cdot z$, as expected.

Lemma~\ref{P:ahs-derivable} can then be adapted to compare relational
structures with the corresponding algebraic ones, including morphisms.
We present only the case for object-free categories.
\begin{proposition}\label{P:of-cat-iso}
  The categories of object-free categories and algebraic object-free
  categories, both with (bounded) morphisms, are isomorphic.
\end{proposition}
\begin{proof}[Sketch of Proof]
  The proof is routine. For objects, relative to
  Lemma~\ref{P:ahs-derivable}, it remains to derive the object-free
  category axioms from their algebraic counterparts.  For morphisms,
  it is easy to check that every bounded morphism in one kind of
  structure is a bounded morphism in the other kind, using the
  definitions of $\cdot$ and $D$ in object-free categories and
  defining $R^x_{yz} \Leftrightarrow x=y\cdot z\land D^x_y$ in their
  algebraic counterpart.  Preservation of units is automatic from
  these results.
\end{proof}

%%%%%%%%%%%%%%%%%%%%%%%%%%%%%%%%%%%%%%%%%%%%%%%%%%%%%%%%%

\section{$\ell r$-Semigroups  and Type-Two Categories}\label{S:lr-cats}

Lemma~\ref{P:rel-props3} translates the definitions of object-free
categories from Section~\ref{S:relational-monoids} into Mac Lane's and
Freyd and Scedrov's alternative form~\cite{MacLane98}, yet more
generally for relational semigroups.

\begin{definition}\label{D:lr-semig}
  A \emph{relational $\ell r$-semigroup} is a structure $(X,R,\ell,r)$ such that
  $\ell,r:X\to X$, $(X,\cdot,R)$ is a relational
  semigroup that satisfies laws (1), (2) and (4) of
  Lemma~\ref{P:rel-props3} and, for all $x,y\in X$,
  \begin{equation*}
    D^x_y\Leftrightarrow r\, x = \ell\, y,\qquad R^y_{(\ell\, x)x}\Rightarrow x = y,  \qquad 
R^y_{x(r\, x)}\Rightarrow x = y.
  \end{equation*}
%  \begin{alignat}{3}
% r\, (\ell\, x) = \ell\, x,&\qquad &
% \ell\, (r\, x) = r\, x,\tag{$\ell r 1$}\label{eq:LR1}\\
% R^x_{(\ell\, x)x},&\qquad &
% R^x_{x(r\, x)},\tag{$\ell r2$}\label{eq:LR2}\\
% R^y_{(\ell\, x)x}\Rightarrow x = y, & \qquad &
% R^y_{x(r\, x)}\Rightarrow x = y,\tag{$\ell r 3$}\label{eq:LR3}\\
% r\, x = \ell\, y \land R^v_{xy} \Rightarrow \ell\, v = \ell\,x,
% &\qquad &
% r\, x = \ell\, y \land R^v_{xy}\Rightarrow r\, v = r\, y.\tag{$\ell r 4$}\label{eq:LR4}
%    % (\exists v.\ R^v_{xy}
%    % \land R^w_{vz}) \Leftrightarrow 
%    % \exists v.\ R^w_{xv} \land R^w_{yz}.\tag{LR7}\label{eq:LR7}
%   \end{alignat}
\end{definition}
\begin{definition}\label{D:lr-cat}
  An \emph{$\ell r$-category} is a partial $\ell r$-semigroup.
\end{definition}
In $\ell r$-categories, the $\ell r$-semigroup axioms reduce further
to $D^x_y \Leftrightarrow r\, x = \ell\, y$ and 
\begin{gather*}
r\, (\ell\, x)  = \ell\, x,\qquad \ell\, (r\, x) = r\, x,\\
  \ell\, x \cdot x = x,\qquad x\cdot r\, x = x,\qquad r\, x = \ell\, y
  \Rightarrow \ell\, (x\cdot y) = \ell\, x,\qquad
  r\, x = \ell\, y \Rightarrow r\, (x\cdot y) = r\, y,\\
  r\, x = \ell\, y \land r\, y= \ell\, z \Rightarrow (x\cdot y) \cdot
  z = x \cdot (y \cdot z).
\end{gather*}
These are essentially MacLane's axioms. A \emph{morphism of relational
  $\ell r$-semirings} is a relational morphism $f:X\to X'$ between
$\ell r$-semirings $X$ and $X'$ that preserves $\ell$ and $r$, that
is, $f\circ \ell = \ell'\circ f$ and $f\circ r= r'\circ f$.

\begin{proposition}
  The categories of object-free categories and $\ell r$-categories,
  both with (bounded)  morphisms, are isomorphic.
\end{proposition}
\begin{proof}
  For objects, it is routine to check that every
  $\ell r$-category is an object-free category and vice versa, using
  that $X_\ell = \{x\mid \ell\, x = x\} = \{x\mid r\, x = x\} = X_r$
  in relational $\ell r$-semigroups $X$ serves as $E$ in relational
  monoids. For morphisms, we know that bounded partial semigroup
  morphisms and their algebraic counterparts are the same. It
  remains to check preservation of $\ell$ and $r$ and preservation of
  units. Let $e=\The y\in E.\ R^x_{yx}$ and $f:X\to X'$ be a bounded
  relational monoid morphism. Then $f\, (\ell\, x) = f\, e$, hence
  ${R'}^{f\, x}_{(f\, e)(f\, x)}$ and $f\, e = f\, (\ell'\, x)$, which
  proves $\ell$-preservation.  The proof of $r$-preservation is
  similar. Conversely, let $f:X\to X'$ be a bounded relational
  $\ell r$-semigroup morphism and let $\ell\, x$ be the unit of
  $x$. Then $R^x_{(\ell\, x)x}$ and therefore
  ${R'}^{f\, x}_{(f\, (\ell\, x))(f\, x)} = {R'}^{f\, x}_{(\ell'\,
    (f\, x))(f\, x)}$, hence the unit $\ell\, x$ has been preserved.
\end{proof}
Analogous results hold already for coherent relational monoids and
$\ell r$-semigroups.

The standard way of axiomatising a category is to use a collection of
objects $O$ and morphisms $M$ with operations $\dom,\cod:M\to O$
associating a source and target object with each morphism, with an
operation $\id:O\to M$ associating a unit morphism with each object
and an operation of composition (pullback) $\circ$ of morphisms, such
that $f\circ g$ is defined whenever $\cod\, f = \dom\, g$.  The
following axioms hold:
\begin{enumerate}
\item $\dom\, (\id\, A) = A = \cod\, (\id\, A)$,
\item $\dom\, (x\circ y) = \dom\, x$ and $\cod\, (x\circ y) = \cod\, y$ whenever
  $\cod\, x = \dom\, y$,
\item $(x\circ y) \circ z = x\circ (y\circ z)$ whenever $\cod\, x = \dom\,
  y$ and $\cod\, y=\dom\, z$,
\item $\id\, (\dom\, x)\circ x = x$ and $x\circ
  \id\, (\cod\, x) = x$.
\end{enumerate}

% \begin{lemma}
%   In every category, 
%   \begin{equation*}
%   \dom\, x = \dom\, y \Leftrightarrow \id\, (\dom\, x) = \id\,
%     (\dom\, y)\qquad\text{ and }\qquad
% \cod\, x = \cod\, y \Leftrightarrow \id\, (\cod\, x) = \id\,
%     (\cod\, y).
%   \end{equation*}
% \end{lemma}
% \begin{proof}
%   The implications ($\Rightarrow$) are trivial. For ($\Leftarrow$)
%   suppose $\id\, (dom\, x) = \id\, (\dom\, y)$. Then
%   \begin{equation*}
% \dom\, x = \dom\, (\id\, (dom\, x)) = \dom\, (\id\, (\dom\,
%   y))=\dom\, y
% \end{equation*}
% by the first axiom and similarly for $\cod$.
% \end{proof}
\begin{proposition}\label{P:cat-lrcat}
  Every small category is an $\ell r$-category and vice versa.
\end{proposition}
\begin{proof}[Proof sketch]
  In any small category define $\ell = \id \circ \dom$ and
  $r=\id\circ \cod$.  It is then routine to check the
  $\ell r$-category axioms. In any $\ell r$-category $X$, let $M=X$
  and $O=X_\ell$. Let $\id$ be the identity map on $X$ restricted to
  the embedding $O\to M$. And let $\dom=\ell$ and $\cod=r$, restricted
  to functions of type $M\to O$. It is then routine to verify the
  category axioms. 
\end{proof}

%%%%%%%%%%%%%%%%%%%%%%%%%%%%%%%%%%%%%%%%%%%%%%%%%%%%%%

\section{Adjunction of  Zero}\label{S:unit-collapse}

A \emph{relational semigroup with zero} is a relational semigroup with
an element $0$ that satisfies $R^0_{0x}$ and $R^x_{x0}$ for every
element $x$. Every relational monoid $X$ has one single relational
unit whenever $D=X\times X$, because then $D^e_{e'}$ holds for all
units $e$ and $e'$, but different units cannot be composed. Moreover,
in every relational semigroup $X$ with zero, $D=X\times X$, because
$D^x_0$ and $D^y_0$ hold and imply $D^x_y$ by relational associativity
($R^0_{x0}\land R^0_{y0}\Leftrightarrow R^0_{00}\land R^0_{xy}$). Thus
every partial monoid with zero, and a fortiori every object-free
category is a monoid.  Adjoining a zero to a partial monoid therefore
makes multiple units disappear. The construction yields a total
semigroup in which the elements formerly known as units satisfy weaker
properties.

An element $e$ is a \emph{weak left unit} of a ternary relation $R$ if
$\exists x.\ R^x_{ex}$ and $\forall x.\ R^x_{ex}\vee \neg D^e_x$. It
is a \emph{weak right unit} if $\exists x.\ R^x_{xe}$ and
$\forall x.\ R^x_{xe}\vee \neg D^x_e$. In a relational semigroup $X_0$
with zero adjoined, $\neg (D_0)^x_y\Leftrightarrow (R_0)^0_{xy}$,
hence the second conditions for left and right units become
$\forall x.\ (R_0 )^x_{ex}\vee (R_0)^0_{ex}$ and
$\forall x.\ (R_0 )^x_{xe}\vee (R_0)^0_{ex}$.  A relational semigroup
with zero is a \emph{weak relational monoid} with zero if for every
$x$ there are weak units $e$ and $e'$ such that $R^x_{ex}$ and
$R^x_{xe'}$. The notion of \emph{weak object-free category} is defined
by analogy.

\begin{lemma}
 Every object-free category can be embedded into a weak coherent
 partial monoid with zero. 
\end{lemma}
\begin{proof}
  Suppose $e$ is a left (right) unit in the object-free category
  $(X,R)$.  We must show that it is a weak left (right) unit in the
  coherent relational semigroup with zero $(X_0,R_0)$.

 Suppose $\exists x.\ R^x_{ex}$ and
  $\forall x,y.\ R^y_{ex} \Rightarrow y = x$. Then obviously
  $\exists x.\ (R_0)^x_{ex}$ since in particular $(R_0)^0_{e0}$ after
  the adjunction. Let $x\in X_0$. If $x=0$, then $(R_0)^0_{ex}$ and
  the second claim holds. If $x\in X$, then either it has left unit
  $e$, in which case $(R_0)^x_{ex}$, or it has another left unit
  $e'\neq e$ in $X$, that is, $R^x_{e'x}$. Then suppose that
  $R^y_{ex}$ for some $y$.  Then there exists a $z$ such that
  $R^y_{zx}$ and $R^x_{ee'}$ by relational associativity, hence
  $D^e_{e'}$, a contradiction. It follows that $\neg D^e_x$ and
  therefore $(R_0)^0_{ex}$ by definition of $R_0$. Hence once again
  the claim holds and the proof is finished. The proof for right units
  is similar. 
\end{proof}

\begin{lemma}
Every weak coherent partial monoid $(X_0,R_0)$ with zero contains the
object-free category $(X,R)$ as a subalgebra.
\end{lemma}
\begin{proof}
  Suppose $e$ is a weak left (right) unit of a semigroup $(S_0,R_0)$
  with zero that is not a unit of $0$. We must show that $e$ is a left
  (right) unit of the partial semigroup $(S,R)$.

  Suppose $\exists x.\ (R_0)^x_{ex}$ and
  $\forall x.\ (R_0)^x_{ex}\vee (R_0)^0_{ex}$. Suppose $x\neq 0$ and
  ${R_0}^x_{ex}$. Then clearly $R^x_{ex}$.  Now let $R^y_{ex}$. Then
  $x=y$ because either $R^x_{ex}$ or $\neg D^e_x$.
\end{proof}

Similar results hold for relational monoids and weak relational
monoids. Example~\ref{ex:weak-units} below shows how several units in a
relational monoid turn into weak units after adjoining a zero.

The adjunction of zero and the resulting collapse can also be
described as follows. Relational monoids over $X$ embed into powerset
quantales in $\Pow\, X$ that are based on complete atomic boolean
algebras by defining the complex product
$A \cdot B = \left\{x\mid \exists y\in A. \exists z\in B.\
  R^x_{yz}\right\}$.
There is a simple Stone-type duality between categories of atomic
quantales and categories of ternary relations~\cite{JonssonT51}.  The
atom structure of the quantale defines the associated ternary
relation; the elements of the ternary relation embed as atoms into the
power set quantale via $\eta\, x = \{x\}$ for all $x\in X$.  In fact,
$\eta$ is an isomorphism, because
$R^x_{yz} \Leftrightarrow \{x\}\subseteq \{y\}\cdot \{z\}
\Leftrightarrow R^{\eta\, x}_{(\eta\, y)(\eta\, z)}$.

If $X$ is a partial monoid, then $\{x\cdot y\} = \{x\}\cdot \{y\}$ for
all $x,y\in X$.  The monoid $(\Pow\, X, \cdot, E)$ thus has the
semigroup with zero $((\Pow\, \eta)\, X,\cdot,\emptyset)$ as a
subsemigroup. It is isomorphic to the semigroup with zero
$(X_0,\cdot_0)$ via $\eta_0$, which extends $\eta$ by
$\eta_0: 0\mapsto \emptyset$.  The elements of $E$ are therefore weak
units in $(\Pow\, \eta)\, X$. The set $E$ is of course not in general
a unit set of this semigroup.

The situation for $\ell r$-semigroups and $\ell r$-categories is much
simpler (see Example~\ref{ex:no-collapse} below).  In fact,
$\ell r$-categories with a zero adjoined have already been introduced by
Schweizer and Sklar. 

\begin{definition}[\cite{SchweizerS67}]\label{D:categorical-semigroup}
  A \emph{categorical semigroup} is a semigroup $(S,\cdot)$ equipped
  with an element $0$ and two functions $\ell,r:S\to S$ that satisfy
  laws (1) and (2) of Lemma~\ref{P:rel-props3} and, for all
  $x,y\in S$,
  \begin{gather*}
r\, 0 = 0\qquad\text{ and } \qquad x\cdot y \neq 0 \Leftrightarrow x\neq 0\land y\neq
0 \land r\, x = \ell\,  y. 
  \end{gather*}
\end{definition}

\begin{proposition}\label{P:LR2cs}
  Every $\ell r$-category $X$ can be embedded into a categorical
  semigroup $X\cup\{0\}$ by adjoining a zero and extending $\ell$ and
  $r$ as $\ell\, 0 = 0$ and $r\, 0=0$.  Every categorical semigroup
  $S$ is an $\ell r$-category with
  $D^x_y\Leftrightarrow r\, x = \ell\, y$. The algebra $S-\{0\}$ forms
  a sub-$\ell r$-category.
\end{proposition}
Categorical semigroups are therefore totalisations of
$\ell r$-semigroups. 

%%%%%%%%%%%%%%%%%%%%%%%%%%%%%%%%%%%%%%%%%%%%%%%%%%%%%%%%%%%%%%

\section{Examples}\label{S:examples}

\begin{lemma}\label{ex:impartial}
  The class of partial monoids is strictly contained in the class of
  relational monoids.
\end{lemma}
\begin{proof} 
Let $a,b$ be letters of some alphabet $\Sigma$ and let
$\mathop{\bowtie}: \Sigma^* \times \Sigma^* \to \Pow\, \Sigma^*$ denote
  the shuffle operation of two words---an example of a
  multioperation. Define $R^x_{yz}\Leftrightarrow x \in y\bowtie z$.
  Obviously $ab,ba\in a \bowtie b$, but $ab\neq ba$. The shuffle 
  monoid $(\Sigma^*, \mathop{\bowtie}, \varepsilon)$ is a total relational monoid and therefore coherent. It has
  one single unit: the empty word $\varepsilon$.
\end{proof}

\begin{lemma}\label{ex:mult-units}
  Relational monoids may have multiple units.
\end{lemma}
\begin{proof}
  The set $X=\{e,e'\}$ with $R^e_{ee}$ and $R^{e'}_{e'e'}$ fixed forms
  a relational monoid with units $e\neq e'$.
\end{proof}

% \begin{example}
%   The category $\mathsf{Rel}$ with sets as objects and binary
%   relations as arrows does not have a zero. Instead it has empty
%   relations of type $X\to X$ for each $X$.\qed
% \end{example}

\begin{lemma}\label{ex:incoherent}
  The class of coherent relational monoids is strictly contained in
  that of relational monoids.
\end{lemma}
\begin{proof}
  The relational monoid $\left(\{a,e\},R\right)$ with $R^e_{ee}$,
  $R^a_{ea}$ and $R^a_{ae}$ and unit $e$ satisfies $D^e_a$, but not
  $D^a_a$. 
\end{proof}

\begin{lemma}\label{ex:pmon-incoherent}
The class of object-free categories is strictly contained in that of
partial monoids. 
\end{lemma}
\begin{proof}
We present two examples of non-coherent partial monoids. 
  \begin{enumerate}
  \item The set $\{a,1\}$ with $D=\left\{(1,1),(a,1),(1,a)\right\}$ and partial
    composition $1\cdot 1 = 1$ and $1\cdot a = a = a\cdot 1$ is a
    partial monoid. As a single-object category it should be a
    monoid---but it is not. The functions $\ell$ and $r$ are forced to
    be $\ell\, a = \ell\, 1 = 1=r\, a = r\, 1$, but
    $r\, a = \ell\, a$ and $\neg D^a_a$ violate coherence
    (Lemma~\ref{P:coherent-rel-monoid}).
    
  \item It is easy to check that partial endofunctions on a set $X$
    form a partial monoid with composition of two partial functions
    $f$ and $g$ defined as $f\cup g$ whenever
    $\dom\, f \cap \dom\, g = \emptyset$ and with the empty partial
    function $\varepsilon$, which satisfies
    $\dom\, \varepsilon = \emptyset$, as its only unit. It is
    straighforward to pick partial functions $f$, $g$, $h$ that
    satisfy
    $\dom\, f \cap \dom\, g = \emptyset = \dom\, g\cap \dom\, h$ and
    $\dom\, f \cap \dom\, h\neq \emptyset$. Then
    $\dom\, (f\cup g) \cap \dom h \neq \emptyset$, and therefore
    $D^f_g$ and $D^g_h$, but not $D^{f\cup g}_h$, which violates
    coherence.\qedhere
  \end{enumerate}
\end{proof}

\begin{lemma}\label{ex:weak-units}
Relational semigroups with zero 
may have several week units.
\end{lemma}
\begin{proof}
  Consider again Example~\ref{ex:mult-units}.  Adjoining a zero forces
  $R^0_{e0}$, $R^0_{0e}$, $R^0_{e'0}$, $R^0_{0e'}$, $R^0_{00}$,
  $R^0_{ee'}$ $R^0_{e'e}$. Now $e$ and $e'$ are no longer units of
  $R_0$. For instance, $e$ is no longer a left unit of $R$ because
  $R^0_{e0}$, $R^0_{ee'}$ and $e'\neq 0$. The other unit cases are
  similar. However, $e$ and $e'$ are weak units after the adjunction.
\end{proof}

\begin{lemma}\label{ex:no-collapse}
$\ell r$-Categories with zero (categorical semigroups) may have
several units. 
\end{lemma}
\begin{proof}
  The categorical semigroup $S=\{x,0\}$ with composition defined by
  $x\cdot x = x$ (the rest is forced) and $\ell$ and $r$ by
  $\ell\, 0 = 0 = r\, 0$ and $\ell x = x = r\, x$ satisfies
  $\ell\, 0\neq \ell\, x$ and $r\, 0 \neq r\, x$. 
\end{proof}

%%%%%%%%%%%%%%%%%%%%%%%%%%%%%%%%%%%%%%%%%%%%%%%%%%%%%%%%%%

\bibliographystyle{alpha}
\bibliography{ofcat}

\end{document}